\documentclass[onecolumn,9pt]{elsarticle}
\usepackage{amssymb}    % symboles ams
\usepackage{amsmath}    % symboles ams maths
\usepackage{amsfonts}
\usepackage{bbm}
\usepackage{mathrsfs}
\usepackage{graphicx}
\usepackage{natbib}
\usepackage{a4wide}
\newtheorem{theorem}{Theorem}[section]

\newtheorem{proposition}{Proposition}[section]
\newtheorem{define}{Definition}[section]
\newtheorem{remark}{Remark}[section]

\newtheorem{problem}{Problem}[section]
\newenvironment{proof}[1]{\textit{Proof} \textit{#1}:\hspace{2mm}}{$\square$\\}

 {\everymath{\displaystyle\everymath{}}\equation}%
 {\endequation}

\DeclareMathOperator*{\diag}{diag}

\providecommand{\Lgain}[3]{$||#1||_{\mathcal{L}_2}<#3||#2||_{\mathcal{L}_2}$}
\providecommand{\myline}{\hbox{\raisebox{0.4em}{\vrule depth 0pt
height 0.4pt width \textwidth}}}
\begin{document}

\begin{frontmatter}

\title{Memory Resilient Gain-scheduled State-Feedback Control of Uncertain LTI/LPV Systems with
Time-Varying Delays}%\thanksref{work}}

%\thanks[work]{A lighter version of this paper has been submitted to IFAC World Congress 2008}

\author[First]{C.~Briat}
\author[Second]{O.~Sename}
\author[Third]{J.F.~Lafay}

%\thanks[First2]{This author did nothing at all}

\address[First]{KTH, School of Mathematics, Division of Optimization and Systems Theory, SE-100 44 Stockholm, SWEDEN, {\tt\small cbriat@math.kth.se, corentin@briat.info, http://www.briat.info}}
\address[Second]{GIPSA-Lab, Departement of Control Systems (former LAG),  Grenoble Universit\'{e}s, ENSE$^3$ - BP46, 38402 Saint Martin d'H\`{e}res - Cedex FRANCE, {\tt\small olivier.sename@gipsa-lab.grenoble-inp.fr, http://www.lag.ensieg.inpg.fr/sename/}}
\address[Third]{IRCCyN - Centrale de Nantes, 1 rue de la No\"{e}  - BP 92101, 44321 Nantes Cedex 3 -
FRANCE, {\tt\small Jean-Francois.Lafay@irccyn.ec-nantes.fr}}

%\begin{keyword}
%    Linear parameter varying systems; Time delay systems; $\mathcal{H}_\infty$ gain scheduled controller; Relaxation
%algorithm; Parameter-dependent linear matrix inequalities
%\end{keyword}


\begin{abstract}
The stabilization of uncertain LTI/LPV
time delay systems with time varying delays by state-feedback controllers is addressed. At the difference of other works in the literature, the proposed approach allows for the synthesis of resilient controllers with respect to
uncertainties on the implemented delay. It is emphasized that such controllers unify memoryless and exact-memory controllers usually considered in the literature. The solutions to the stability and stabilization problems are expressed in terms of LMIs which allow to check the stability of the closed-loop system for a given bound on the knowledge error and even optimize the uncertainty radius under some performance constraints; in this paper, the $\mathcal{H}_\infty$ performance measure is considered. The interest of the approach is finally illustrated through several examples.
\end{abstract}

\end{frontmatter}

\section{Introduction}

Since several years, time-delay systems \citep{Moon:01,
Niculescu:01, Zhang:01a, GuKC:03, Han:05, Gouaisbaut:06, Fridman:06,
Suplin:06, Kao:07, Ariba:09} have attracted more and more interest since they arise in various problems \citep{KolmaMy:99,
Niculescu:01} such as chemical processes, biological systems,
economic systems, etc. The presence of delays, in the
equations describing the process, is often responsible of
destabilizing effects and performance deterioration. Indeed, in fast
systems, even a small time-delay may have a very harmful effect, and
thus cannot be neglected. This has motivated the development of many
types of stability tests and matched controller design techniques
\citep{Zhong:06, Fridman:06a, XuLam:06, Briat:07, Briat:08a}. Nevertheless, while the theories for stability analysis and stabilization of LTI systems with constant delays are well established, the case of time-varying delays is still not well understood \cite{Louisell:99, Papa:07}.

On the other hand, over the past recent years, LPV systems
\citep{Apkarian:95a, Apkarian:98a, Wu:01, Scherer:01, Iwasaki:01a}
have been heavily studied since they offer a very general approach for the modeling
and the control of complex systems such as nonlinear systems. This
fresh upsurge of gain-scheduling based techniques is mainly due to
the emergence of LMIs \citep{Boyd:94a}, which provide a powerful
formalism for the expression of solutions of many problems arising
in systems and control theory. It is important to note that many
problems in LPV framework remain open and major improvements
are still expected for stability analysis as well as for control law
synthesis.

The stability analysis of LPV time-delay systems is still an open
problem and is discussed for instance in \citep{Zhang:02a, Zhang:05, Briat:08c, Briat:08phd} while the
control of LPV time-delay systems have also been studied in
\citep{Wu:01a, Zhang:02a, Zhang:05, Briat:08phd, Briat:09f,  Briat:09e} but still
remains sporadic. These systems belong to the intersection of two
families and hence inherit from the difficulties of each
one. Additionally, new difficulties arise, for instance several
robust control tools which are used to deal with LPV systems (such
as the projection lemma \cite{Gahinet:94a} or the dualization lemma \citep{Scherer:01}) are
difficult to apply to LPV time-delay systems. Indeed, the stability analysis of such systems cannot be tackled using classical Lyapunov functions (as for finite dimensional LPV systems) but must be analyzed using adapted tools, namely Lyapunov-Razumikhin functions and Lyapunov-Krasovskii functionals, which increase the number of decision variables.

The main contribution of the paper concerns
the synthesis of non-fragile controllers with respect to an
uncertain knowledge of the implemented delay. Several papers mention
non-fragility of observers/controllers but the robust stability
analysis is done only after the synthesis \citep{VerriestSP:02,
Sename:06}. In such a case, the non-fragility radius (maximal tolerable delay
uncertainty) is difficult to guarantee or optimize. In the proposed
approach, the non-fragility radius can be fixed by the designer or
even maximized. Moreover, since the approach is based on LMIs,
time-varying and uncertain systems can be also easily handled at the
difference of latter results which were based on frequency domain
methods, thus restricted to LTI systems and difficult to generalize
to the uncertain case. A close result but notably different is also provided in \citep{Ivanescu:00} where the control of a time-delay system with constant delay is performed using a controller with a different delay. However, no relationship between the delays is considered and thus the maximal admissible error cannot be analyzed from this result. Results on resilience of controllers with delay uncertainties are also provided in \citep{Briat:09e, Briat:09f} in the framework on \textit{delay-scheduled controllers} where the delay acts as a gain-scheduling parameter on the controller expression. In the current paper, the initial time-delay system system structure for both the system and the controller is kept, in order to develop a stabilization result in this domain. The resulting problem can be equivalently represented as a stabilization problem of a system with two delays where the delays are coupled through an algebraic inequality. Thus the problem reduces in a correct and efficient accounting of this inequality in the LMI-based stabilization result.

It is worth mentioning that almost all the works from the literature address the stabilization problem with memoryless (conservative) or exact-memory (non-implementable) controllers. The approach of the paper is more pragmatic and aims at designing controllers whose delay is different from the system one. Delay estimation techniques \cite{Belkoura:08} could be applied in order to determine the delay to implement in the controller. Such a class of controllers finds applications in the control of physical systems with state delays, such as the ones in \cite{KolmaMy:99}.

The goal of this article is not to provide better results on the
stability of time-delay systems by introducing new
Lyapunov-Krasovskii functionals but is to show that it is possible
to consider uncertainties of the delay knowledge and take it into
account in an efficient way in the synthesis. The proposed approach is very general and
can be extended to many types of (less conservative) different
functionals.

For a real square matrix $M$ we define $M^\mathcal{S}:=M+M^T$ where $M^T$ is its transpose. The space of
signals with finite energy is denoted by $\mathcal{L}_2$ and the
energy of $v\in\mathcal{L}_2$ is
$||v||_{\mathcal{L}_2}:=\left(\int_{0}^{+\infty}v^*(t)v(t)dt\right)^{1/2}$.
$\mathbb{S}^k_{++}$ denotes the cone of real symmetric positive
definite matrices of dimension $k$. $\mathbb{R}_+$
($\mathbb{R}_{++}$) denotes the set of nonnegative (positive) real
numbers. $\star$ denotes symmetric terms in symmetric matrices and
in quadratic forms. $\otimes$ and $\times$ denote the Kronecker and the cartesian product respectively. $co\{S\}$ stands for the convex hull of the set $S$. $col_i(\lambda_i)$ is the column vector with components $\lambda_i$.

% PAPER OUTLINE

The paper is structured as follows. In Section \ref{sec:def},
definitions and objectives of the paper are defined. In Section
\ref{sec:stab} we provide several delay-dependent stability results
for uncertain LPV time-delay systems with time-varying delays.
Section \ref{sec:sf} is devoted to the development of constructive
sufficient conditions to the existence of three types of parameter
dependent state-feedback controllers. Finally, in Section
\ref{sec:ex}, examples and discussions on the proposed approach
are provided.

\section{Definitions and Objectives}\label{sec:def}

%\subsection{Definitions}

The following class of systems will be considered in the paper:
\begin{equation}\label{eq:syst}
  \begin{array}{lcl}
    \dot{x}(t) &=& A(\lambda,\rho)x(t)+A_h(\lambda,\rho)x_h(t)+B(\lambda,\rho)u(t)\\
    &&+E(\lambda,\rho)w(t)\\
    z(t) &=& C(\lambda,\rho)x(t)+C_h(\lambda,\rho)x_h(t)+D(\lambda,\rho)u(t)\\
    &&+F(\lambda,\rho)w(t)\\
    x(\eta) &=& \phi(\eta),\ \eta\in[-h_M,0]
  \end{array}
\end{equation}
where $x\in\mathbb{R}^n$, $x_h(t)=x(t-h(t))\in\mathbb{R}^n$,
$u\in\mathbb{R}^m$, $w\in\mathbb{R}^p$, $z\in\mathbb{R}^q$ and
$\phi(\cdot)$ are respectively the system state, the delayed state,
the control input, the exogenous input, the controlled output and
the functional initial condition. The system
matrices are defined by
\begin{equation}\label{eq:systpol}
  \begin{bmatrix}
    A & A_h & B & E\\
    C & C_h & D & F\\
  \end{bmatrix}(\lambda,\rho)=\sum_{i=1}^N\lambda_i\begin{bmatrix}
    A_i & A_{hi} & B_i & E_i\\
    C_i & C_{hi} & D_i & F_i\\
  \end{bmatrix}(\rho)
\end{equation}
where $\lambda=col_i(\lambda_i)$ is time-invariant and belongs to
the unitary simplex $\Lambda$ defined by
\begin{equation}
  \Lambda:=\left\{\lambda_i\ge0:\
  \sum_{i=1}^N\lambda_i=1,\ i=1,\ldots,N\right\}
\end{equation}
The delay is assumed to belong to the set
\begin{equation}\label{eq:hset}
\mathscr{H}:=\left\{h:\mathbb{R}_+\to[0,h_M],h\le\mu<1\right\}
\end{equation}
with $h_M<+\infty$. The vector of parameters $\rho$ belongs to
\begin{equation}\label{eq:pset}
\mathscr{P}:=\left\{\rho:\mathbb{R}_+\to
U_\rho\subset\mathbb{R}^{N_p}, \dot{\rho}\in co\{U_\nu\}\right\}
\end{equation}
where $N_p>0$ is the number of parameters, $U_\rho$ is a connected
compact subset of $\mathbb{R}^{N_p}$. $U_\nu$ is the set of vertices
of the convex set in which the derivative of the
parameters evolves and is defined by
\begin{equation}\label{eq:Unu}
U_\nu:=\times_{i=1}^{N_p}\left\{\underline{\nu}_i,\bar{\nu}_i\right\}
\end{equation}
where $\underline{\nu}_i$ and $\bar{\nu}_i$ denote respectively the
upper and lower bound of $\dot{\rho}_i$; hence we have
$\dot{\rho}\in co\{U_\nu\}$.

The aim of the current paper is to find a control law based on a
parameter dependent state-feedback of the form
\begin{equation}\label{eq:sf}
u(t)=K_0(\rho)x(t)+K_h(\rho)x(t-d(t))
\end{equation}
where the gains $K_0(\rho)$ and $K_h(\rho)$ are sought such that the
controller stabilizes the uncertain LPV system (\ref{eq:syst}). Note that the delay $d(t)$ involved in the
control law is allowed to be different from the system delay $h(t)$.
First, the usual case $d(t)=h(t)$ will be considered and then the
more general case $d(t)=h(t)+\eta(t)$ with $|\eta(t)|\le\delta$ will
be solved in turn. To this aim, the following set is introduced
\begin{equation}\label{eq:dset}
\mathscr{D}_\delta:=\left\{d:\mathbb{R}_+\to[0,h_M],|d(t)-h(t)|\le\delta,\
h\in\mathscr{H}\right\}
\end{equation}
and defines the set of controller delays.

%\begin{remark}
%The derivative of the implemented delay is not considered since it
%is assumed here that $d(t)$ may be very noisy (thus
%nondifferentiable). However, if the implemented delay is chosen to
%be constant (e.g. $d(t)=h_M/2$ which is the value minimizing the error bound to $h_M/2$) then it would be less conservative to take into
%account the derivative of $d(t)$ which is zero. In such a case, the
%derivative of the error $h(t)-d(t)$ satisfies
%$|\dot{h}(t)-\dot{d}(t)|\le\mu<1$. This case is a direct extension of the
%proposed approach and will be not developed.
%\end{remark}

\begin{define}
  The following terminology is used for the controllers (\ref{eq:sf}):
  \begin{itemize}
    \item If $K_h(\cdot)=0$ the controller is referred to as a \textbf{memoryless
    controller};
    \item If $d(t)=h(t)$ for all $t\ge0$ (i.e. $\delta=0$ in (\ref{eq:dset})) then the controller
    is referred to as an \textbf{exact memory controller};
    \item If $|d(t)-h(t)|\le\delta$ for some $\delta>0$ then the controller is referred to as a
    \textbf{$\delta$-memory resilient controller}.
  \end{itemize}
\end{define}
The set $\mathscr{D}_\delta$ is parameterized by the uncertainty
radius $\delta\ge0$. Note that when $\delta=0$ the equality
$d(t)=h(t)$ holds for all $t\ge0$ and hence the $\delta$-memory
resilient controller coincides with exact memory controller. Note
also that if $\delta=h_M$ then the implemented delay may take any
value inside $[0,h_M]$ independently of the value of $h(t)$. In such a case, $h(t)$ can be considered as unknown, resulting then in the
particular case where the memoryless and the $h_M$-memory
resilient controllers are actually quite near, structurally speaking.
It will be illustrated in the examples that $\delta$-memory
resilient controllers connect together the well-known memoryless and
exact-memory controllers by providing a unique and generalized
expression for all controllers.

With this in mind, it is possible to state the main problem of the paper:
%\begin{problem}\label{pb:1}
%  Find a memoryless or exact-memory parameter dependent state-feedback controller (\ref{eq:sf}) with $d(t)=h(t)$ which
%  \begin{enumerate}
%    \item Robustly asymptotically stabilizes system (\ref{eq:syst}): $x(t)\rightarrow 0$ as $t\rightarrow +\infty$ with $w(t)=0$, for all parameter trajectories
%    $\rho\in\mathscr{P}$, for all delay $h\in\mathscr{H}$ and for all $\lambda\in\Lambda$.
%    \item Provides a guaranteed $\mathcal{L}_2$ performance attenuation gain from $w$ to $z$ satisfying $||z||_{\mathcal{L}_2}<\gamma||w||_{\mathcal{L}_2}$ with $x(\eta)=0,\
%    \eta\in[-h_M,0]$, $w(t)\ne0$ for all parameter trajectories
%    $\rho\in\mathscr{P}$, for all delay $h\in\mathscr{H}$ and for all $\lambda\in\Lambda$.
%  \end{enumerate}
%\end{problem}
\begin{problem}\label{pb:2}
  Find a parameter dependent $\delta$-memory resilient state-feedback controller (\ref{eq:sf}) which
  \begin{enumerate}
     \item Robustly asymptotically stabilizes system (\ref{eq:syst}): $x(t)\rightarrow 0$ as $t\rightarrow +\infty$ with $w(t)=0$, for all parameter trajectories
    $\rho\in\mathscr{P}$, for all delays $(h,d)\in\mathscr{H}\times\mathscr{D}_\delta$ and for all
    $\lambda\in\Lambda$.
    \item Provides a guaranteed $\mathcal{L}_2$ performance attenuation gain from $w$ to $z$ satisfying $||z||_{\mathcal{L}_2}<\gamma||w||_{\mathcal{L}_2}$ with $x(\eta)=0,\
    \eta\in[-h_M,0]$, $w(t)\ne0$ and for all parameter trajectories
    $\rho\in\mathscr{P}$, for all delays $(h,d)\in\mathscr{H}\times\mathscr{D}_\delta$ and for all
    $\lambda\in\Lambda$.
  \end{enumerate}
\end{problem}
%
%Problem \ref{pb:1} can be viewed as a 'nominal' problem where it is
%assumed that either the delay is unknown or the delay is perfectly
%known. In the former case, the control law reduces to
%$u(t)=K_0(\rho)x(t)$ while the latter one considers control laws
%$u(t)=K_0(\rho)x(t)+K_h(\rho)x(t-h(t))$. However, controllers
%involving an exact delay-value are quite difficult to implement due
%to the strong assumption of exact delay knowledge. Moreover, it is a
%well known fact that the estimation of delays is a very difficult
%problem \citep{Drakunov:06, Belkoura:08}. Hence, the control law
%(\ref{eq:sf}) with exact delay value is not acceptable from a
%practical point of view.
%
%The second problem tries to remedy this problem by allowing for a
%difference between the values of implemented delay and the system
%delay. Moreover, the stability and $\mathcal{L}_2$ input/output
%performance are guaranteed for variations of the implemented delay
%within given bounds around the actual delay. This approach makes the
%implementation of control laws of the form (\ref{eq:sf}) possible
%since, for the first time, resilience with respect to delay
%uncertainty is provided in terms of LMIs.
%%It is important to note
%%that robustness with respect to delay uncertainty was generally made
%%after the synthesis \citep{VerriestSP:02, Sename:06} in LMI-based
%%approaches. Moreover the approach provided in this paper can be
%%applied for constant delays as well as for time-varying delays, LTIV
%%systems, etc\dots

%
\section{A Control Oriented Delay-Dependent Stability Result}\label{sec:stab}
This section is devoted to the stability analysis of LPV time-delay
systems of the form (\ref{eq:syst}). The results are based on the
extension of \citep{Han:05, Gouaisbaut:06} to uncertain LPV
time-delay systems with time-varying delays. As we shall see later,
the immediate LMI conditions derived from the Lyapunov-Krasovskii
theorem are not well suited for stability and synthesis problems due
to:
\begin{itemize}
  \item multiple products between system
matrices and decision variables; and
  \item quadratic terms on the polytope
variable $\lambda$.
\end{itemize}
The second part of the proof is then devoted to the relaxation of
such a result in order to both reduce the number of these products (to
one) and make the dependence on the polytope variable $\lambda$
affine. This makes the derivation of stabilization result easier and
overall more efficient than approaches based on relaxations made
after substitution of the closed-loop system matrices into the LMI.

Another approach based on the
projection lemma \citep{Apkarian:95a} and on adjoint systems
\citep{Bensoussan:06} was considered in \citep{Briat:08c}. This
approach led to an equivalent problem independent of the
controller matrices involving a nonlinear matrix inequality which was
then solved using an iterative LMI algorithm. The current
approach avoids such a computational complexity by tolerating an increase of conservatism. So, following this new idea, the following stability analysis result is stated:
\begin{theorem}\label{th:th1ter}
  System (\ref{eq:syst}) with no control input (ie. $u(t)\equiv0$) is robustly asymptotically
   stable for all
  $(h,\rho,\lambda)\in\mathscr{H}\times \mathscr{P}\times\Lambda$ if there
  exist continuously differentiable matrix functions
  $P_i:U_\rho\rightarrow \mathbb{S}^{n}_{++}$ for all $i\in\{1,\ldots,N\}$, a matrix function $X:U_\rho\to\mathbb{R}^{n\times n}$, $N$ constant matrices
  $Q_i,R_i\in\mathbb{S}^{n}_{++}$ and a scalar $\gamma>0$ such that the parameter dependent
  LMIs
    \begin{equation}\label{eq:th1ter1}
        \Theta_i=\begin{bmatrix}
            -X(\rho)^\mathcal{S} & \Phi_{12i} & \Phi_{13i} & \Phi_{14i} & 0 & X(\rho)^T & h_MR_i\\
            \star & \Phi_{22i} & R_i & 0 & C_i(\rho)^{T} & 0 & 0\\
            \star  & \star & \Phi_{33i} & 0 & C_{hi}(\rho)^{T} & 0 & 0\\
            \star& \star & \star & -\gamma I_p & F_i(\rho)^T & 0 & 0\\
            \star & \star & \star & \star & -\gamma I_q & 0 & 0\\
            \star & \star & \star & \star & \star & -P_i(\rho) & -h_MR_i\\
            \star & \star & \star & \star & \star & \star & -R_i\\
        \end{bmatrix}\prec0
    \end{equation}
hold for all $(\rho,\nu,i)\in U_\rho\times
co\{U_\nu\}\times\{1,\ldots,N\}$ with
\begin{equation*}
  \Phi_{22i}=\dfrac{\partial
  P_i(\rho)}{\partial\rho}\nu-P_i(\rho)+Q_i-R_i
\end{equation*}
\begin{equation*}
\begin{array}{lclclcl}
    \Phi_{12i}&=&P_i(\rho)+X(\rho)^TA_i(\rho)&& \Phi_{13i}&=&X(\rho)^TA_{hi}(\rho)\\
    \Phi_{14i}&=&X(\rho)^TE_i(\rho) && \Phi_{33i}&=&-(1-\mu)Q_i-R_i
    \end{array}
\end{equation*}
In such a case, system (\ref{eq:syst}) satisfies
\Lgain{z}{w}{\gamma} for all $(\rho,h,\lambda)\in\mathscr{P}\times
\mathscr{H}\times\Lambda$.
\end{theorem}
\begin{proof}{}
The choice of the functional $V$ is inspired from \citep{Gouaisbaut:06, Han:05} and
extended to the case of LPV time-delay systems with time-varying
delays as in \cite{Briat:08c} and the supply rate $s(w,z)$ is considered:
\begin{equation}\label{eq:LKF}
\begin{array}{rcl}
    V&=&\displaystyle{x^T(t)P(\lambda,\rho)x(t)+\int_{t-h(t)}^tx^T(\theta)Q(\lambda)x(\theta)d\theta}\\
    &&+h_M\displaystyle{\int_{-h_M}^0\int_{t+\theta}^tx^T(\eta)R(\lambda)x(\eta)d\eta
    d\theta}\\
    s(w(t),z(t))&=&\gamma w^T(t)w(t)-\gamma^{-1} z(t)^Tz(t)
\end{array}
\end{equation}
with $P(\lambda,\rho)=\sum_{i=1}^N\lambda_iP_i(\rho)$,
$Q(\lambda)=\sum_{i=1}^N\lambda_iQ_i$,
$R(\lambda)=\sum_{i=1}^N\lambda_iR_i$. The supply-rate $s(w,z)$
characterizes the $\mathcal{L}_2$-gain from $w$ to $z$. Define the function $H$ to be
\begin{equation}
  H=V-\int_0^ts(w(\theta),z(\theta))d\theta
\end{equation}
%
%If $\dot{H}$, computed along the trajectories solutions of system
%(\ref{eq:syst}), is negative definite for all $\lambda\in\Lambda$,
%$\rho\in\mathscr{P}$, $h(t)\in[0,h_M]$ and $\dot{h}(t)\le\mu$ then according to the dissipativity theory
%\cite{Willems:72, Scherer:05a}, system (\ref{eq:syst}) is robustly
%asymptotically stable and satisfies $||z||_{\mathcal{L}_2}<\gamma
%||w||_{\mathcal{L}_2}$.
%
The derivative of the function can be bounded from above by
\begin{equation}
  \begin{array}{lcl}
  \dot{H} &\le&x^T(t)\displaystyle{\partial  P(\lambda,\rho)\over \partial
  \rho}\dot{\rho}(t)x(t)+x(t)^T\left[A^T(\lambda,\rho)P(\lambda,\rho)\right]^{\mathcal{S}}x(t)\\
  &&+2x_h^T(t)A_{h}(\lambda,\rho)^TP(\lambda,\rho)x(t)+2w^T(t)E(\lambda,\rho)^TP(\lambda,\rho)x(t)\\
  &&+x(t)^TQ(\lambda)x(t)-(1-\dot{h})x_h(t)^TQ(\lambda)x_h(t)\\
  &&+h_M^2\dot{x}(t)^TR(\lambda)\dot{x}(t)+\mathcal{I}-\gamma w(t)^Tw(t)\\
  &&+\gamma^{-1}z(t)^Tz(t)\\
  z(t)&=&C(\lambda,\rho)x(t)+C_h(\lambda,\rho)x(t-h(t))+F(\lambda,\rho)w(t)\\
  \mathcal{I}&=&-h_M\displaystyle{\int_{t-h(t)}^t\dot{x}(\theta)^TR\dot{x}(\theta)d\theta}
  \end{array}
\end{equation}
Note that $ -(1-\dot{h})\le -(1-\mu)$ and using Jensen's inequality
\citep{GuKC:03} on $\mathcal{I}$ we obtain
\begin{equation}
    \mathcal{I} \le
    -\left(\int_{t-h(t)}^t\dot{x}(s)ds\right)^TR(\lambda)(\star)^T\label{eq:I}
\end{equation}
Then expanding the expression of $s(w,z)$ we get LMI (\ref{eq:th1}) after Schur complements.
%
%Then expanding the expression of $s(w,z)$ and performing two
%successive Schur complements w.r.t. the terms
%$$\left[C(\lambda,\rho)x(t)+C_h(\lambda,\rho)x_h(t)+E(\lambda,\rho)w(t)\right]^T\gamma^{-1}\left[\star\right]$$
%and
%$$\left[A(\lambda,\rho)x(t)+A_h(\lambda,\rho)x(t-h(t))+F(\lambda,\rho)w(t)\right]^T(h_M^2R(\lambda))\left[\star\right]$$
%we get LMI (\ref{eq:th1}).
%
\begin{figure*}
\begin{equation}\label{eq:th1}
  \begin{bmatrix}
    [A(\lambda,\rho)^TP(\lambda,\rho)]^\mathcal{S}+Q(\lambda)-R(\lambda) &  P(\lambda,\rho)A(\lambda,\rho)+R(\lambda) & P(\lambda,\rho)E(\lambda,\rho) &  C(\lambda,\rho)^T & h_MA(\lambda,\rho)^TR(\lambda)\\
    \star & -(1-\mu)Q(\lambda)-R(\lambda) & 0 & C_h(\lambda,\rho)^T  & h_MA_h(\lambda,\rho)^TR(\lambda)\\
    \star & \star & -\gamma I_p & F(\lambda,\rho)^T &
    h_ME(\lambda,\rho)^TR(\lambda)\\
    \star & \star & \star & -\gamma I_q & 0\\
        \star & \star & \star & \star & -R(\lambda)\\
  \end{bmatrix}\prec0
\end{equation}
\myline
\end{figure*}
Due to products between decision matrices and system matrices, it is
clear that the LMI is not linear in $\lambda$ (e.g.
$A(\lambda,\rho)P(\lambda,\rho)$). Moreover, the structure of
(\ref{eq:th1}) is not adapted to the controller design due to the
presence of multiple products terms $PA,PA_h,RA$ and $RA_h$
preventing to find a linearizing change of variable even after
congruence transformations.  A relaxed version of (\ref{eq:th1}) is
then expected to remove the multiple products preventing the change
of variables and limit the increase of conservatism. A similar
approach has been used in \citep{Tuan:01, Tuan:03, Briat:08phd}.

In view of relaxing the latter result, define $\Theta=\sum_{i=1}^N\lambda_i\Theta_i$ where $\Theta_i$ is given in
(\ref{eq:th1ter1}). Below, it is proved that $\Theta\prec0$ implies the
feasibility of (\ref{eq:th1}). First note that $\Theta$
can be written as (where the dependency on $\lambda,\rho$ and
$\dot{\rho}$ are dropped for clarity):
\begin{equation}\label{eq:djsqljdl12313}
  \left.\Theta\right|_{X=0}+Z_1^TXZ_2+Z_2^TX^TZ_1\prec0
\end{equation}
where $Z_1=\begin{bmatrix}   -I & A & A_h & E & 0 & 0
\end{bmatrix}$ and $Z_2=\begin{bmatrix}   I & 0 & 0 & 0 & 0 & 0
\end{bmatrix}$. Then invoking the projection lemma
\citep{Gahinet:94a}, the feasibility of $\Theta\prec0$ implies the feasibility of the underlying LMI problem
\begin{subequations}\label{eq:djsqljdl123132}
  \begin{gather}
    \mathcal{N}_1^T\left.\Theta\right|_{X=0}\mathcal{N}_1\prec0\label{eq:djsqljdl123132_1}\\
    \mathcal{N}_2^T\left.\Theta\right|_{X=0}\mathcal{N}_2\prec0\label{eq:djsqljdl123132_2}
  \end{gather}
\end{subequations}
where $\mathcal{N}_1$ and $\mathcal{N}_2$ are basis
of the null-space of $Z_1$ and $Z_2$ respectively. Note that since $X$ only depends on $\rho$ (and neither on $\dot{\rho}$ nor $\lambda$) then equivalence equivalence between (\ref{eq:djsqljdl12313}) and (\ref{eq:djsqljdl123132}) is lost and reduces then to an implication from (\ref{eq:djsqljdl12313}) to (\ref{eq:djsqljdl123132}) only.

After some tedious manipulations, it
is possible to show that LMI (\ref{eq:djsqljdl123132_1}) is equivalent to
(\ref{eq:th1}) and thus that $\Theta\prec0$ implies (\ref{eq:th1}).
Thus the conservatism of the
approach is characterized by LMI (\ref{eq:djsqljdl123132_2}) and by the sole dependence of $X$ on $\rho$.
%
%Finally, using standard arguments on convexity to remove the terms
%$\lambda_i$, we obtain Theorem \ref{th:th1ter}.
\end{proof}

LMIs (\ref{eq:th1ter1}) do not involve any multiple products and hence can be easily used for design purpose. The removal of
multiple products has been allowed through the introduction of a
'slack' variable $X(\rho)$. The additional conservatism and the
(slight) increase of the computational complexity are the price to
pay to get easily tractable conditions for the stabilization
problem.%

\section{Delay-Dependent Stabilization by State-Feedback}\label{sec:sf}

%This section is devoted to the control of time-delay systems of the
%form (\ref{eq:syst})-(\ref{eq:systpol}) with a state-feedback
%controller of general form (\ref{eq:sf}). The three types of
%controllers will be considered in this section.

%Let us recall that when the delay is unknown, the particular case of
%controller with $K_h(\cdot)=0$ will be considered and whenever it is
%exactly known then the delay of the controller $d(t)$ will be
%considered as identical to $h(t)$. However, due to practical
%difficulties to have a thorough knowledge of the delay value, the
%case $d(t)=h(t)+\eta(t)$ with $|\eta(t)|\le\delta$ is more relevant
%in practice and will be addressed in the second part of the section.

\subsection{Robust Stabilization using Memoryless and Exact-Memory State-Feedback
Controllers}\label{sec:nom}

In this part, stabilizing control laws of the form
\begin{equation}\label{eq:sf1}
  u(t)=K_0(\rho)x(t)+K_h(\rho)x(t-h(t))
\end{equation}
are sought. The closed loop system obtained from the interconnection of system
(\ref{eq:syst}) and controller (\ref{eq:sf1}) is given by
\begin{equation}\label{eq:cl1}
\begin{array}{lcl}
      \dot{x}(t) &=& A_{cl}(\lambda,\rho)x(t)+A_{hcl}(\lambda,\rho)x_h(t)+E(\lambda,\rho)w(t)\\
      z(t)&=& C_{cl}(\lambda,\rho)x(t)+C_{hcl}(\lambda,\rho)x_h(t)+F(\lambda,\rho)w(t)
\end{array}
\end{equation}
with
$A_{cl}(\lambda,\rho)=A(\lambda,\rho)+B(\lambda,\rho)K_0(\rho)$,
$C_{cl}(\lambda,\rho)=C(\lambda,\rho)+D(\lambda,\rho)K_0(\rho)$,
$A_{hcl}(\lambda,\rho)=A_h(\lambda,\rho)+B(\lambda,\rho)K_h(\rho)$
and
$C_{hcl}(\lambda,\rho)=C_h(\lambda,\rho)+D(\lambda,\rho)K_h(\rho)$.
The following theorem on robust stabilization is obtained:
\begin{theorem}\label{th:th2_1}
There exists a parameter dependent state-feedback control of the
form (\ref{eq:sf1}) which robustly asymptotically stabilizes system
(\ref{eq:syst}) for all $(h,\rho,\lambda)\in\mathscr{H}\times
\mathscr{P}\times\Lambda$ if there exist continuously differentiable
matrix functions $P_i:U_\rho\rightarrow \mathbb{S}^{n}_{++}$,
constant matrices $Q_i,R_i\in\mathbb{S}^{n}_{++}$ for
$i=1,\ldots,N$, $X\in\mathbb{R}^{n\times n}$, matrix functions
$Y_0,Y_h:U_\rho\rightarrow \mathbb{R}^{m\times n}$ and a scalar
$\gamma>0$ such that the parameter dependent LMIs
        \begin{equation}\label{eq:th2_1}
        \begin{bmatrix}
            \Xi_{11i} & \Xi_{12i} & \Xi_{13i} & E_i(\rho) & 0 & X & h_MR_i\\
            \star & \Xi_{22i} & R & 0 & \Xi_{24i} & 0 & 0\\
            \star  & \star & \Xi_{33i} & 0 & \Xi_{34i} & 0 & 0\\
            \star &  \star & \star & -\gamma I_p & F_i(\rho)^T & 0 & 0\\
            \star & \star & \star & \star & -\gamma I_q & 0 & 0\\
            \star & \star & \star & \star & \star & -P_i(\rho) & -h_MR_i\\
            \star & \star & \star & \star & \star & \star & -R_i\\
        \end{bmatrix}\prec0
    \end{equation}
hold for all $(\rho,\nu,i)\in U_\rho\times
co\{U_\nu\}\times\{1,\ldots,N\}$ where
\begin{equation*}
  \begin{array}{lclclcl}
    \Xi_{11}&=&-X^\mathcal{S}\\
    \Xi_{12i}&=&P_i(\rho)+A_i(\rho)X+B_i(\rho)Y_0(\rho)\\
    \Xi_{23i}&=&A_{hi}(\rho)X+B_i(\rho)Y_h(\rho)\\
    \Xi_{22i}&=&\dfrac{\partial P_i(\rho)}{\partial\rho}\nu-P_i(\rho)+Q_i-R_i\\
    \Xi_{33i}&=&-(1-\mu)Q_i-R_i\\
    \Xi_{24i}&=&[C_i(\rho)X+D_i(\rho)Y_0(\rho)]^T\\
    \Xi_{34i}&=&[C_{hi}(\rho)X+D_i(\rho)Y_h(\rho)]^T\\
  \end{array}
\end{equation*}
In such a case, a suitable control law is given by (\ref{eq:sf1})
with gains $K_0(\rho)=Y_0(\rho)X^{-1}$ and
$K_h(\rho)=Y_h(\rho)X^{-1}$. Moreover, the closed-loop system
satisfies \Lgain{z}{w}{\gamma} for all
$(h,\rho,\lambda)\in\mathscr{H}\times \mathscr{P}\times\Lambda$.
\end{theorem}
\begin{proof}{}
 Substitute the closed-loop system (\ref{eq:cl1}) into inequality
(\ref{eq:th1ter1}) and set $X$ to be a constant matrix. $X$ is enforced to be constant
in order to allow for the use of congruence
transformations, otherwise, nonlinear terms would appear, making the solution to the problem difficult to solve (i.e. non LMI). Then performing a congruence transformation with respect to matrix
$$\diag(I_3\otimes X^{-1},I_{p+q},I_2\otimes X^{-1})$$ and applying the
following linearizing change of variable
\begin{equation*}
  \begin{array}{lclclcl}
    X&\leftarrow&X^{-1} &       &  P_i& \leftarrow&X^{-T}P_iX^{-1}\\
    Q_i&\leftarrow&X^{-T}Q_iX^{-1} &   & R_i&\leftarrow&X^{-T}R_iX^{-1}\\
    Y_0&\leftarrow&K_0X^{-1} &  &   Y_h&\leftarrow&K_hX^{-1}
  \end{array}
\end{equation*}
yields LMIs (\ref{eq:th2_1}).
\end{proof}
\begin{remark}
The latter result can be used to both synthesize exact-memory and
memoryless control laws. Memoryless structures can be obtained
setting $Y_h(\cdot)=0$.
\end{remark}
%It is important to mention that these results provide a solution to
%both memoryless and with memory state-feedback controllers.
%Memoryless controller can be designed by setting $Y_h(\cdot)=0$ in
%the LMI condition. This solution is preferred when no information is
%available in real time on the delay value.
%
%On the other hand, it is clear from the expression of the control
%law (\ref{eq:sf1}) and the open-loop system (\ref{eq:syst}) that the
%exact knowledge of the delay value is crucial for the validity of
%the control law. Moreover, the estimation or measurement of the
%delay in a state-delayed system is difficult and hence control law
%(\ref{eq:cl1}) may not be practically valid \citep{Apkarian:98a}.
%
%
\subsection{Robust Stabilization using $\delta$-Memory Resilient State-Feedback Controllers}
This part establishes a new result on the stabilization of time-delay systems where the strong constraint on exact delay knowledge is relaxed. In the following, we will show that it is also possible to give stabilization conditions even in presence on time-varying uncertainties on the delay knowledge. We will see, as a direct consequence of the method, that Theorem \ref{th:th2_1} can be retrieved from Theorem \ref{th:th3_2}, the main result of this section.

In the sequel, the following control law will be considered:
\begin{equation}\label{eq:sf2}
  u(t)=K_0(\rho)x(t)+K_h(\rho)x(t-d(t))
\end{equation}
where the approximate value of the delay $d(t)\in\mathscr{D}_\delta$
is used. To the authors' knowledge, the only work on such control laws using LMI techniques is \citep{Ivanescu:00}. However, the provided approach only considers constant time-delays and does not consider any relationship between the delays. In the current approach, time-varying delays are allowed and the delay knowledge maximal error is explicitly taken into account in the stabilization conditions. Moreover, with the provided approach, it is easy to guarantee a given bound on the error or optimize it, at the difference of \cite{VerriestSP:02,
Sename:06}.

The closed-loop system given by the interconnection of the control
law (\ref{eq:sf2}) and system (\ref{eq:syst}) is governed by the
expressions:
\begin{equation}\label{eq:cl2}
  \begin{array}{lcl}
    \dot{x}(t)&=&A_{cl}(\lambda,\rho)x(t)+A_h(\lambda,\rho)x_h(t)\\
    &&+B(\lambda,\rho)K_h(\rho)x_d(t)+E(\lambda,\rho)w(t)\\
    z(t)&=&C_{cl}(\lambda,\rho)x(t)+C_h(\lambda,\rho)x_h(t)\\
    &&+D(\lambda,\rho)K_h(\rho)x_d(t)+F(\lambda,\rho)w(t)
  \end{array}
\end{equation}
where $x_h(t)=x(t-h(t))$, $x_d(t)=x(t-d(t))$,
$A_{cl}(\lambda,\rho)=A(\lambda,\rho)+B(\lambda,\rho)K_0(\rho)$ and
$C_{cl}(\lambda,\rho)=C(\lambda,\rho)+D(\lambda,\rho)K_0(\rho)$.
It is worth noting that this system is not a classical system with
two delays. Indeed, the difficulty lies in the fact that both delays
satisfy an algebraic inequality which constrains the trajectories of
$d(t)$ to evolve within a ball, of radius $\delta$, centered around the trajectory of
$h(t)$. This additional information has to be taken into account for an efficient stability and performance analysis of the system
(\ref{eq:cl2}). To this aim, the following preliminary result from \cite{Shustin:07} is used:
\begin{proposition}
Let us define the operator $\Delta(\cdot)$ as
  \begin{equation}
  \begin{array}{lcl}
        \Delta(z_0(t))&=&\dfrac{2}{\sqrt{7}\delta}\displaystyle{\int_{t-d(t)}^{t-h(t)}z_0(\tau)d\tau}\quad \mathrm{with\ }(h,d)\in\mathscr{H}\times\mathscr{D}_\delta
  \end{array}
\end{equation}
For any input signal $\xi\in\mathcal{L}_2$, the output $\Delta(\xi)$ is also in $\mathcal{L}_2$ and we have
$||\Delta(\xi)||_{\mathcal{L}_2}\le||\xi||_{\mathcal{L}_2}$.
\end{proposition}
Using this operator, we can
turn system (\ref{eq:cl2}) into an uncertain single-delay system
(i.e. with $d(t)$ only) depending explicitly on the delay error
bound $\delta$:
\begin{equation}\label{eq:cl3}
  \begin{array}{lcl}
    \dot{x}(t)&=&A_{cl}(\lambda,\rho)x(t)+A_{hcl}(\lambda,\rho)x_d(t)\\
    &&+E(\lambda,\rho)w(t)+\frac{\sqrt{7}}{2}\delta A_h(\lambda,\rho)w_0(t)\\
    z(t)&=&C_{cl}(\lambda,\rho)x(t)+C_{hcl}(\lambda,\rho)x_d(t)\\
    &&+F(\lambda,\rho)w(t)+\frac{\sqrt{7}}{2}\delta C_h(\lambda,\rho)x_d(t)\\
    z_0(t)&=&\dot{x}(t)\\
    w_0(t)&=&\Delta(z_0(t))
  \end{array}
\end{equation}
where $x_d(t)=x(t-d(t))$,
$A_{cl}(\lambda,\rho)=A(\lambda,\rho)+B(\lambda,\rho)K_0(\rho)$,
$A_{hcl}(\lambda,\rho)=A_h(\lambda,\rho)+B(\lambda,\rho)K_h(\rho)$,
$C_{cl}(\lambda,\rho)=C(\lambda,\rho)+D(\lambda,\rho)K_0(\rho)$ and
$C_{hcl}(\lambda,\rho)=C_h(\lambda,\rho)+D(\lambda,\rho)K_h(\rho)$.
%
%It is important to note that this uncertain single-delay system
%(\ref{eq:cl3}) where the uncertainty radius $\delta$ appears
%explicitly in the expression of the transformed system. Hence
%(\ref{eq:cl3}) is a comparison system for (\ref{eq:cl2}) in which
%the delays satisfy $|d(t)-h(t)|\le\delta$.

%\begin{remark}
%    The use of the operator $\Delta(\cdot)$ introduces additional dynamics in system (\ref{eq:cl3}) as
%shown in \citep{GuKC:03}, hence system (\ref{eq:cl3}) is not
%equivalent to system (\ref{eq:cl2}). The analysis of the
%conservatism is performed for constant time-delay only since the
%approach of \citep{GuKC:03} is based on frequency-domain methods. A Lyapunov-Krasovskii approach
%is also provided in \citep{Kharitonov:03}. The presence of additional dynamics emphasizes that
%the model (\ref{eq:cl3}) is a non-equivalent comparison model for
%(\ref{eq:cl2}), that is, the stability of (\ref{eq:cl3}) is a
%sufficient condition for the stability of (\ref{eq:cl2}) only.
%%
%%An important result states that the conservatism only depends on the
%%eigenvalues of the matrix $A_h$ and if the matrix $A_h$ is Hurwitz
%%then both systems are equivalent. In the stabilization problem, the
%%matrix acting on the delayed state is 'controlled': $A_h+BK_h$ and
%%hence can be made Hurwitz by an appropriate choice of the matrix
%%$K_h$ (assuming that the pair $(A_h,B)$ is stabilizable). Hence, the
%%impact of additional dynamics may be less critical in the
%%stabilization problem than in the stability analysis problem.
%\end{remark}
%
Finally, according to the previous discussion, the main result of the paper is given below:
\begin{theorem}\label{th:th3_2}
There exists a parameter dependent state-feedback control of the
form (\ref{eq:sf2}) which robustly asymptotically stabilizes system
(\ref{eq:systpol}) for all
$(h,d,\rho,\lambda)\in\mathscr{H}\times\mathscr{D}_\delta\times\mathscr{P}\times\Lambda$
if there exist continuously differentiable matrix functions
$P_i:U_\rho\rightarrow \mathbb{S}^{n}_{++}$, matrix functions $S_i:U_\rho\rightarrow \mathbb{S}^{n}_{++}$, constant matrices
$Q_i,R_i\in\mathbb{S}^{n}_{++}$, $X\in\mathbb{R}^{n\times n}$,
$Y_0,Y_h:U_\rho\rightarrow \mathbb{R}^{m\times n}$ and a scalar
$\gamma>0$ such that the parameter dependent LMIs
\begin{equation}\label{eq:rtyuio31}
  \begin{bmatrix}
    \Omega_{1i} & \Omega_{2i}\\
    \star & \Omega_{3i}
    \end{bmatrix}\prec0
\end{equation}
hold for all $(\rho,\nu,i)\in U_\rho\times
co\{U_\nu\}\times\{1,\ldots,N\}$ and where
\begin{equation*}
\begin{array}{lcl}
  \Omega_{1i}&=&\begin{bmatrix}
    \Omega_{11i} &  \Omega_{12i} & \Omega_{13i} & \Omega_{14i}& E_i(\rho) & 0\\
    \star & \Omega_{22i} & R_i & \frac{\sqrt{7}}{2}\delta R_i & 0 & \Omega_{26i}\\
    \star & \star & \Omega_{33i} & \frac{\sqrt{7}}{2}\delta\Omega_{33i} & 0 & \Omega_{36i}\\
    \star & \star & \star & \Omega_{44i} & 0 & \frac{\sqrt{7}}{2}\delta C_{hi}(\rho)^T\\
    \star & \star & \star & \star & -\gamma I_p & F_i(\rho)^T\\
    \star & \star & \star & \star & \star & -\gamma I_q\\
  \end{bmatrix}\\
  \Omega_{2i}&=&\begin{bmatrix}
    X & S_i(\rho) & h_MR_i\\
     0 & 0 & 0\\
     0 & 0 & 0\\
     0 & 0 & 0\\
     0 & 0 & 0\\
     0 & 0 & 0
  \end{bmatrix}\\
  \Omega_{3i}&=&\begin{bmatrix}
    -P_i(\rho) & -S_i(\rho) & -h_MR_i\\
\star & -S_i(\rho) & 0\\
\star & \star & -R_i
  \end{bmatrix}\\
    \Omega_{11i}&=&-X^\mathcal{S}\\
    \Omega_{12i}&=&P_i(\rho)+A_i(\rho)X+B_i(\rho)Y_0(\rho)\\
    \Omega_{13i}&=&A_{hi}(\rho)X+B_i(\rho)Y_h(\rho)\\
    \Omega_{14i}&=&\dfrac{\sqrt{7}}{2}\delta A_{hi}(\rho)X\\
    \Omega_{22i}&=&\dfrac{\partial P_i}{\partial \rho}\nu-P_i(\rho)+Q_i-R_i\\
    \Omega_{33i}&=&-(1-\mu)Q_i-R_i\\
    \Omega_{26i}&=&(C_i(\rho)X+D_i(\rho)Y_0(\rho))^T\\
    \Omega_{36i}&=&(C_{hi}(\rho)X+D_i(\rho)Y_h(\rho))^T\\
    \Omega_{44i}&=&\frac{7}{4}\delta^2\Omega_{33i}-S_i(\rho)
    \end{array}
  \end{equation*}
In such a case, the controller gains can be computed using
    $K_0(\rho)=L_0(\rho)X^{-1}$ and $K_h(\rho)=L_h(\rho)X^{-1}$. Moreover the closed-loop system satisfies \Lgain{z}{w}{\gamma} for
all $(\rho,h,d,\lambda)\in\mathscr{P}\times
\mathscr{H}\times\mathscr{D}_\delta\times\Lambda$.
\end{theorem}

\begin{proof}{}
 The proof is given in \ref{ap:th2_2}.
\end{proof}

\begin{remark}\label{rem:smooth}
Here, the derivative of the error is not restricted and is allowed to be arbitrarily large. If for some reason, the derivative is bounded from above by one, the latter results can be refined. Indeed, a sharper bound on $||\Delta||_{\mathcal{L}_2-\mathcal{L}_2}$ can be shown to be 1. Moreover, the derivative bound can also be taken into account through the introduction of another operator (the delay operator) similarly as in \cite{GuKC:03}.
%  It must be pointed out that the error on the delay knowledge is
%supposed to have arbitrarily fast variation rate since it is
%considered, for instance, that the estimated delay is noisy.
%However, it is possible to use a constant delay $d$ in the
%controller and in this case, it would be less conservative to
%consider also the rate of variation of the error $\eta(t)=d-h(t)$
%given by $\dot{\eta}(t)=\dot{h}(t)$. This may be done
%using the following operator:
%\begin{equation}
%  \nabla(\zeta(t-h(t)))=\sqrt{1-\mu}\zeta(t-d)
%\end{equation}
%satisfying
%$||\nabla(\xi)||_{\mathcal{L}_2}\le||\xi||_{\mathcal{L}_2}$
%\citep{GuKC:03}. Moreover, under the assumption of a constant delay $d$, the term $2/\sqrt{7}$ of the operator $\Delta(\cdot)$ can be set to 1 according to the result in \citep{GuKC:03}.
\end{remark}

It is important to note that when the delay is exactly known (i.e. $\delta=0$), then Theorem \ref{th:th3_2} reduces to Theorem \ref{th:th2_1}. This is stated in the following proposition:
\begin{proposition}\label{rq:gen}
When $\delta=0$, then LMIs
(\ref{eq:rtyuio31}) is equivalent to LMIs (\ref{eq:th2_1}) provided that the matrix $S(\lambda,\rho)\succ0$
is chosen sufficiently small (e.g. according to the 2-norm).
\end{proposition}
\begin{proof}{}
The proof is only sketched since it relies on simple arguments and easy calculations. First, set $\delta=0$ in (\ref{eq:rtyuio31}), this creates 0 entries on the 4$^{th}$ row and column of (\ref{eq:rtyuio31}) except for the diagonal value which is $-S_i(\rho)$. Since $S_i(\rho)$ is positive definite we can remove the 4$^{th}$ row and column from (\ref{eq:rtyuio31}). Now, we just need to analyze the impact of the remaining terms depending on $S_i(\rho)$ which are located on the 8$^{th}$ row and column of (\ref{eq:rtyuio31}). A Schur complement on the block $(8,8)$ leads to a matrix sum of the form $\Lambda_i(\rho)+\Upsilon_i(\rho)\prec0$ where $\Lambda_i(\rho)$ is exactly LMI (\ref{eq:th2_1}) and
$$\Upsilon_i(\rho)=\begin{bmatrix}
  S_i(\rho) & 0 & \ldots & 0 & -S_i(\rho) & 0\\
  0 & 0 & \ldots & 0 & 0 & 0\\
  \vdots & \vdots & \ddots & \vdots & \vdots & \vdots\\
  0 & 0 & \ldots & 0 & 0 & 0\\
  -S_i(\rho) & 0 & \ldots & 0 & S_i(\rho) & 0\\
  0 & 0 & \ldots & 0 & 0 & 0
\end{bmatrix}$$
It can be shown that this matrix has $n$ positive eigenvalues (those of $S_i(\rho)$) and $\alpha-n$ zero eigenvalues where $\alpha$ is the dimension of $\Upsilon_i(\rho)$. Thus, choosing $S_i(\rho)$ as small as necessary, it is possible to approximate arbitrarily well LMI (\ref{eq:th2_1}) by LMI (\ref{eq:rtyuio31}) with $\delta=0$.
\end{proof}

This shows that the main result embeds naturally (by construction) the case of controller with exact memory. Another important fact is when $\delta=h_M$, we get a result which is close to the memoryless case. This will be illustrated in the examples.

\begin{remark}
  The LMI conditions of Theorems \ref{th:th1ter}, \ref{th:th2_1} and \ref{th:th3_2}  must be satisfied for
all $(\rho,\nu)\in U_\rho\times co\{U_\nu\}$. However, it is
possible to reduce the computational complexity through the reduction of parameter set from $U_\rho\times co\{U_\nu\}$ to $U_\rho\times U_\nu$. This is done using the following particular structure for the matrix
$P(\lambda,\rho)=P_0(\rho)+\sum_{i=1}^N\lambda_iP_i$. In this case,
the resulting LMI conditions involve no product between $\lambda$ and $\dot{\rho}$ and hence LMIs have to be checked only on the set $U_\rho\times co\{U_\nu\}$.
\end{remark}

\section{Examples}\label{sec:ex}

This section is devoted to examples and discussions on the
provided approach. It will be illustrated that the current approach
improves result of the literature in the control of LPV time-delay
systems. Moreover, the connection between memoryless and
exact-memory controllers through $\delta$-memory resilient
controllers will be illustrated. Let us consider the following system borrowed from \citep{Wu:01a}
and modified in \citep{Zhang:05}:
\begin{equation}\label{eq:exsyst}
\begin{array}{lcl}
  \dot{x}(t)&=&\begin{bmatrix}
    0 & 1+\phi\sin(t)\\
    -2 & -3+\sigma\sin(t)
  \end{bmatrix}x(t)+\begin{bmatrix}
    \phi\sin(t) & 0.1\\
    -0.2+\sigma\sin(t) & -0.3
  \end{bmatrix}x(t-h(t))\\
  &&+\begin{bmatrix}
    0.2\\0.2
  \end{bmatrix}w(t)+\begin{bmatrix}
    \phi\sin(t)\\
    0.1+\sigma\sin(t)
  \end{bmatrix}u(t)\\
  z(t) &=& \begin{bmatrix}
    0 & 10\\
    0 & 0
  \end{bmatrix}x(t)+\begin{bmatrix}
    0\\0.1
  \end{bmatrix}u(t)
\end{array}
\end{equation}
where $\phi=0.2$ and $\sigma=0.1$. Define $\rho(t):=\sin(t)$, $h_M=0.5$ and $\mu=0.5$ as in \citep{Zhang:05}.

\subsection{Example 1: Memoryless State-Feedback (Small delay)}

First, a memoryless control law is computed using Theorem \ref{th:th2_1}.
The parameter dependent decision matrices are chosen to be
\begin{equation*}
\begin{array}{lcl}
      P(\rho)&=&P^0+P^1\rho+P^2\rho^2\\
      Y_0(\rho)&=&Y_0^0+Y_0^1\rho+Y_0^2\rho^2\\
\end{array}
\end{equation*}
Verifying the LMI of Theorem  \ref{th:th2_1} over a grid of
$N_g=41$ points yields a minimal value $\gamma^*=1.9089$ which is
better than all results obtained before. In \citep{Zhang:05}, a
minimal value of $\gamma=3.09$ is found while in \citep{Briat:08c}
an optimal value of $\gamma=2.27$ is obtained (using a nonlinear
approach). The controller computed using Theorem  \ref{th:th2_1} is
given by
$$K_0(\rho)=\begin{bmatrix}
  -1.0535-2.9459\rho+1.9889\rho^2\\
  -1.1378-2.6403\rho+1.9260\rho^2
\end{bmatrix}^T$$
It is worth noting that the results are even better while the
controller has smaller coefficients than in the other approaches
\citep{Zhang:05, Briat:08c}. It is hence expected to have a smaller
control input which should remain within acceptable bounds, even in
presence of disturbances.

The influence of $\mu$ on the delay-margin (with $\mathcal{L}_2$
performance constraint) is detailed in Table \ref{tab:hvsmu} where,
as expected, the delay margin decreases as the value of $\mu$
increases. Moreover, the results of \citep{Zhang:05} are more
conservative than those obtained using Theorem  \ref{th:th2_1} since the
delay-margin is always smaller than $1.4$ (for $\gamma<10$) for any value of $\mu$.

As a final remark, the maximal value for $h_M$ obtained for $\mu=0$
is large and this suggests that the system might be
delay-independent stabilizable in the constant delay case.

\begin{table*}
%\begin{table}\label{tab:hvsmu}
  \begin{center}
    \begin{tabular}{l|c|c|c|c}
        $\mu$ & 0 & 0.5 & 0.9 & 0.99\\
        \hline
        \hline
        $h_M$ \citep{Zhang:05} & 1.2 & 1.2 & 1.2 & $\sim1$\\
        $h_M$ \citep{Briat:08c} & -- & 79.1511 & -- & --\\
        $h_M$ Theorem \ref{th:th2_1} & 929.1372 & 371.0928 & 6.9218 & 2.9325
    \end{tabular}
    \caption{Evolution of the delay margin $h_M$ with respect to the
    bound on the delay derivative $\mu$ for a maximal allowable $\gamma<10$}\label{tab:hvsmu}
  \end{center}
%\end{table}
\myline
\end{table*}

\subsection{Example 2: State Feedback with Exact Memory (Large delay)}

Still considering system (\ref{eq:exsyst}) but with $h_M=10$ and
$\mu=0.9$, an instantaneous state-feedback controller of
the form $u(t)=K_0(\rho)x(t)$ is sought. Theorem \ref{th:th2_1} yields:
\begin{equation}\label{eq:ml}
K_0(\rho)=\begin{bmatrix} 0.5724-6.3679\rho
-1.4898\rho^2\\
-0.7141-4.1617\rho-0.8425\rho^2
\end{bmatrix}^T
\end{equation}
which ensures a closed-loop $\mathcal{L}_2$ input/output performance
level of $12.8799$.
Now an exact-memory state-feedback control law
$u(t)=K_0(\rho)x(t)+K_h(\rho)x(t-h(t))$ is computed using Theorem
\ref{th:th2_1} and the obtained controller gains are given by
\begin{equation*}
  \begin{array}{lcl}
   K_0(\rho)&=&\begin{bmatrix}
    1.0524-2.8794\rho-0.4854\rho^2\\
    -0.7731-1.8859\rho+0.1181\rho^2
  \end{bmatrix}^T\\
    K_h(\rho)&=&\begin{bmatrix}
      -0.6909+0.5811\rho+0.1122\rho^2\\
      -0.0835+0.3153\rho+0.0689\rho^2
    \end{bmatrix}^T
  \end{array}.
\end{equation*}
This controller ensures a closed-loop $\mathcal{L}_2$ performance level of
$4.1641$. The gain of performance resulting from the use of a
controller involving a delayed term is evident. However, this controller is non implementable due to the practical impossibility of knowing the exact delay value at any time. This motivates the synthesis of a memory-resilient controller.

\subsection{Example 3: $\delta$-Memory Resilient Controller
Synthesis (Large delay)}

Finally, a $\delta$-memory resilient state-feedback controller stabilizing system (\ref{eq:exsyst}) is sought.
Using Theorem \ref{th:th3_2}, the achieved minimal closed-loop
$\mathcal{L}_2$ performance with respect to $\delta$ is plotted in
Figure \ref{fig:gamma_vs_delta}. As expected the minimal
$\mathcal{L}_2$ performance level grows as the delay
uncertainty radius increases. This illustrates that the achieved closed-loop
performance is deteriorated when the delay is badly known.

Moreover, there are two remarkable values for the worst-case
$\mathcal{L}_2$ gain, respectively obtained for $\delta=0$ and
$\delta=h_M$. For these particular values we have:
\begin{equation}
  \begin{array}{lcl}
    \left.\gamma\right|_{\delta=0}&=&4.1658\\
    \left.K_0(\rho)\right|_{\delta=0}&=&\begin{bmatrix}
      1.0542-2.8895\rho-0.4827\rho^2\\
      -0.7714-1.8912\rho+0.1216\rho^2
    \end{bmatrix}^T\\
    \left.K_h(\rho)\right|_{\delta=0}&=&\begin{bmatrix}
      -0.6885+0.5849\rho+0.1116\rho^2\\
      -0.0817+0.3148\rho+0.0667\rho^2
    \end{bmatrix}^T
  \end{array}
\end{equation}
\begin{equation}\label{eq:delta10}
  \begin{array}{lcl}
    \left.\gamma\right|_{\delta=10}&=&13.0604\\
     \left.K_0(\rho)\right|_{\delta=10}&=&\begin{bmatrix}
      0.4422-6.3469\rho-1.3619\rho^2\\
     -0.9475-4.1219\rho-0.6140\rho^2
    \end{bmatrix}^T\\
    \left.K_h(\rho)\right|_{\delta=10}&=&\begin{bmatrix}
     -0.0163-0.0005\rho+0.0127\rho^2\\
     -0.0007-0.0006\rho+0.0011\rho^2
    \end{bmatrix}^T
  \end{array}
\end{equation}
When the delay is exactly known (i.e. $\delta=0$), the
$\mathcal{L}_2$ performance index and the controller are very close
(quite identical) to the ones obtained using Theorem
\ref{th:th2_1} which considers exact-memory controllers. This
illustrates well Remark \ref{rq:gen}.

On the other hand, when $\delta=h_M$, it could be considered that the
delay is actually unknown since for any value for $h(t)$, the
implemented value $d(t)$ may take any value into $[0,h_M]$. Hence,
this means that the results with such a value for $\delta$ should be
close to the results obtained using a memoryless control law.
Comparing (\ref{eq:delta10}) and (\ref{eq:ml}) it is possible to
remark that the obtained closed-loop performance level is very near to the one obtained with a memoryless control law. Moreover,  the matrix gain $K_h(\rho)$ in (\ref{eq:delta10}) has a small
norm making it almost identical to a memoryless controller. The delayed action is highly penalized due to a too
large uncertainty on the delay knowledge. The results are summarized
in Table \ref{tab:comp}.
\begin{table}
  \begin{tabular}{l|l||l|l}
    Exact memory & $\gamma=4.1641$ & Memoryless & $\gamma=12.8799$\\
    \hline
    0-resilient & $\gamma=4.1658$ & 10-resilient & $\gamma=13.0604$
  \end{tabular}
  \caption{Comparison of the results obtained using Theorem \ref{th:th3_2}}\label{tab:comp}
\end{table}

The above discussion illustrates that $\delta$-memory resilient controllers define
the intermediary behavior of the closed-loop system between the two
extremal controllers: the memoryless and the exact-memory
controllers. The emphasis of this continuity between memoryless and
exact memory controller through $\delta$-memory resilient
controllers constitutes one of the main contribution of the
paper. Moreover, such controllers are also more realistic, from a practical point of view, than exact-memory controllers.
\begin{figure}
\begin{center}
    \includegraphics[width=0.6\textwidth]{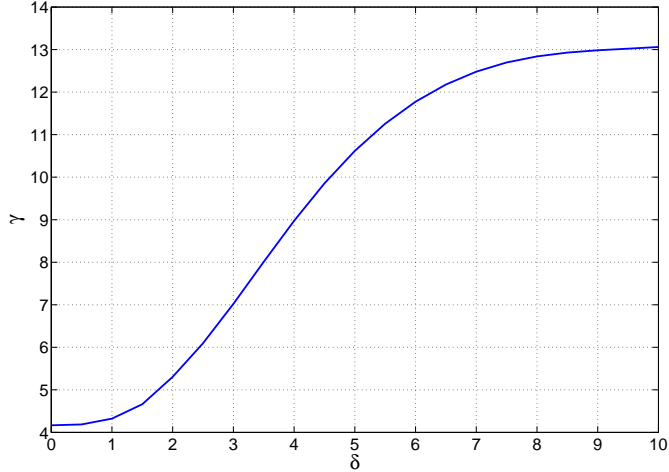}
    \caption{Evolution of the worst-case $\mathcal{L}_2$ gain of the closed-loop system with respect to maximal delay uncertainty $\delta$}\label{fig:gamma_vs_delta}
    \end{center}
\end{figure}
\section{Conclusion}

The current paper introduces a new approach to the stabilization of
LPV time-delay systems using parameter dependent state-feedback
controllers. First, a delay-dependent stability test with
$\mathcal{L}_2$ performance analysis for LPV time-delay systems with
time-varying delays is provided in terms of parameter dependent
LMIs. This result is obtained from the use of a parameter dependent
Lyapunov-Krasovskii functional used along with the Jensen's
inequality, an approach which has proven its efficiency. Since this
result is not well suited for design purpose due to multiple
products between decision matrices and system matrices, a relaxed
version of the result is developed. This version involves an
additional 'slack' (or 'lifting') variable and avoids any nonlinear terms (multiple
products). This allows to find a linearizing change of variable for
the stabilization problem.

A first stabilization result is then provided and characterizes both
memoryless and exact-memory controllers. However, due to the
difficulty of estimating delays, latter controllers are generally
non-implementable in practice. This has motivated the development of
another type of controllers, called '$\delta$-memory resilient controllers' where the delay implementation error is taken into account in the design. It turns out that this new class of controllers connect memoryless and exact-memory
controllers together, through a unique formulation. Indeed, by acting the error bound $\delta$ between the delays, it is possible to recover exact-memory ($\delta=0$), $\delta$-memory resilient ($\delta\in(0,h_M)$) and memoryless ($\delta\sim h_M$) control laws successively. This part constitutes the main contribution of the paper and is illustrated in the examples.

\appendix

\section{Proof of theorem \ref{th:th3_2}}\label{ap:th2_2}

%%%%%%%%%%%%%%%%%%%%%%%%%%%%%%%%%%%%%%%%%%%%%%%%%%%%%%%%%%%%%%%%%%%%%%%%%%%
%%%%%%%%%%%%%%%%%%%%%%%%%%%%%%%%%%%%%%%%%%%%%%%%%%%%%%%%%%%%%%%%%%%%%%%%%%%
\begin{figure*}
  \begin{equation}\label{eq:rtyuio}
  \begin{bmatrix}
    A_{cl}^TP+PA_{cl}+\dot{P}+Q-R & PA_{hcl}+R &\frac{\sqrt{7}}{2}\delta PA_h+\delta R & PE & C_{cl}^T & h_MA_{cl}^TR & A_{cl}^TL\\
    \star & -Q_\mu-R & -\sqrt{2}\delta(Q_\mu+R) & 0 & C_{hcl}^T & h_MA_{hcl}^TR & A_{hcl}^TL\\
    \star & \star & -\frac{7}{4}\delta^2(Q_\mu+R)-L & 0 & \frac{\sqrt{7}}{2}\delta C_h^T & \frac{\sqrt{7}}{2}h_M\delta A_h^TR & \frac{\sqrt{7}}{2}\delta A_h^TL\\
     \star & \star & \star & -\gamma I_p & F^T & h_ME^TR & E^TL\\
     \star & \star & \star & \star & -\gamma I_q & 0 & 0\\
     \star & \star & \star & \star & \star & -R & 0\\
    \star & \star & \star & \star & \star & \star & -L
    \end{bmatrix}\prec0
  \end{equation}
\end{figure*}
\begin{figure*}
  \begin{equation}\label{eq:rtyuio2}
  \begin{bmatrix}
    -(X+X^T) & P+X^TA_{cl} & X^TA_{hcl} & \frac{\sqrt{7}}{2}\delta X^TA_h & X^TE & 0 & X^T & L & h_MR\\
    \star & \dot{P}-P+Q-R & R & \frac{\sqrt{7}}{2}\delta R & 0 & C_{cl}^T & 0 & 0 & 0\\
    \star & \star & -Q_\mu-R & -\frac{\sqrt{7}}{2}\delta(Q_\mu+R) & 0 & C_{hcl}^T & 0 & 0 & 0\\
     \star & \star & \star & -\frac{7}{4}\delta^2(Q_\mu+R)-L & 0 & \frac{\sqrt{7}}{2}\delta C_h^T & 0 & 0 & 0\\
     \star & \star & \star & \star & -\gamma I_p & F^T & 0 & 0 & 0\\
     \star & \star & \star & \star & \star & -\gamma I_q & 0 & 0 & 0\\
     \star & \star & \star & \star & \star & \star & -P & -L & -h_MR\\
    \star & \star & \star & \star & \star & \star & \star & -L & 0\\
    \star & \star & \star & \star & \star & \star & \star & \star & -R\\
    \end{bmatrix}\prec0
  \end{equation}
\end{figure*}
%%%%%%%%%%%%%%%%%%%%%%%%%%%%%%%%%%%%%%%%%%%%%%%%%%%%%%%%%%%%%%%%%%%%%%%%%%%
%%%%%%%%%%%%%%%%%%%%%%%%%%%%%%%%%%%%%%%%%%%%%%%%%%%%%%%%%%%%%%%%%%%%%%%%%%%

Let us consider the closed-loop system (\ref{eq:cl2}) and the
Lyapunov-Krasovskii functional V given in (\ref{eq:LKF}). Computing the
derivative of (\ref{eq:LKF}) along the trajectories solutions of
system  (\ref{eq:cl2}) we get
\begin{equation}
\begin{array}{lcl}
      \dot{V}&\le& X(t)^T\Psi(\lambda,\dot{\rho},\rho)X(t)\\
      \Psi(\lambda,\rho,\dot{\rho})&=&\Pi(\lambda,\rho,\dot{\rho})+h_M^2\Gamma_1(\lambda,\rho)^TR(\lambda)\Gamma_1(\lambda,\rho)
\end{array}
\end{equation}
where $X(t)=col(x(t),x_h(t),x_d(t),w(t))$, $\Pi$ is defined by
\begin{equation*}\label{eq:pi11}
  \begin{bmatrix}
    \Pi_{11} & \Pi_{12} & \Pi_{13} & P(\lambda,\rho)E(\lambda,\rho)\\
    \star & -(1-\mu)Q(\lambda)-R(\lambda) & 0 & 0\\
    \star & \star & 0 & 0\\
    \star & \star & \star & 0
  \end{bmatrix}
\end{equation*}
with
\begin{equation*}
  \begin{array}{lcl}
    \Pi_{11}&=&\left[A_{cl}(\lambda,\rho)^TP(\lambda,\rho)\right]^\mathcal{S}+\dfrac{\partial P(\lambda,\rho)}{\partial \rho}\dot{\rho}+Q(\lambda)-R(\lambda)\\
    \Pi_{12}&=&P(\lambda,\rho)A_h(\lambda,\rho)+R(\lambda)\\
    \Pi_{13}&=&P(\lambda,\rho)B(\lambda,\rho)K_h(\rho)
  \end{array}
\end{equation*}
and $\Gamma_1(\rho)=\begin{bmatrix}A_{cl}(\rho) & A_h(\rho) &
B(\rho)K_h(\rho) & E(\rho)
\end{bmatrix}$.
Now according to the relation
$w_0=\Delta(\dot{x})=2(x_h(t)-x_d(t))/\sqrt{7}\delta$ we have
\begin{equation}
  X(t)=\underbrace{\begin{bmatrix}
    I & 0 & 0 & 0\\
    0 & I & \frac{\sqrt{7}\delta}{2} I & 0\\
    0 & I & 0 & 0\\
    0 & 0 & 0 & I
  \end{bmatrix}}_{\mbox{$M$}}\underbrace{\begin{bmatrix}
    x(t)\\
    x_d(t)\\
    w_0(t)\\
    w(t)
  \end{bmatrix}}_{\mbox{$Y(t)$}}
\end{equation}
and thus
\begin{equation}\label{eq:djsqkodjqsl}
  \dot{V}\le Y(t)^TM^T\Psi(\lambda,\dot{\rho},\rho)MY(t)
\end{equation}
In order to consider the uncertain norm-bounded operator $\Delta$
and input/output $\mathcal{L}_2$ performance, the following
quadratic supply-rate
$s(\zeta(t))=\zeta(t)^T\Phi(\lambda,\rho)\zeta(t)$ defined by
\begin{equation}
    \Phi(\lambda,\rho)=\diag(-L(\lambda,\rho),L(\lambda,\rho),-\gamma I_p,\gamma I_q)
\end{equation}
is added to (\ref{eq:djsqkodjqsl}) where $\zeta(t)=col(w_0(t),z_0(t),w(t),z(t))$ and $w_0(t)=\Delta[z_0(t)]$. Here,
$L(\lambda,\rho)=L(\lambda,\rho)^T\succ0$ is a parameter dependent D-scaling
and $\gamma>0$ characterizes the upper bound on the $\mathcal{L}_2$ gain of the transfer $w\to z$. This leads
to inequality (\ref{eq:rtyuio}) where the dependence on $\rho$ and
$\lambda$ has been dropped for clarity. Then applying the same
relaxation procedure as in the proof of Theorem \ref{th:th1ter} we
get the new inequality (\ref{eq:rtyuio2}) with
$A_{hcl}(\cdot)=A_h(\cdot)+B(\cdot)K_h(\cdot)$. Finally performing a
congruence transformation with respect to
$$diag(I_4\otimes X^{-1},I_{p+q},I_3\otimes X^{-1})$$
along with the change of variables
\begin{equation}
\begin{array}{lclclcl}
X&\leftarrow & X^{-1} & &P_i(\rho)&\leftarrow & X^{-T}P_i(\rho)X^{-1}\\
Q_i&\leftarrow & X^{-T}Q_iX^{-1} && R&\leftarrow & X^{-T}R_iX^{-1}\\
S_i(\rho)&\leftarrow & X^{-T}S_i(\rho)X^{-1} && Y_0(\rho)&\leftarrow & K_0(\rho)X^{-1}\\
Y_h(\rho)&\leftarrow & K_h(\rho)X^{-1}
\end{array}
\end{equation}
yields LMI (\ref{eq:rtyuio31}). This concludes the proof.

\bibliographystyle{plain}
%\bibliography{../../../Lastbib/global}

\begin{thebibliography}{0}
\expandafter\ifx\csname natexlab\endcsname\relax\def\natexlab#1{#1}\fi
\expandafter\ifx\csname url\endcsname\relax
  \def\url#1{\texttt{#1}}\fi
\expandafter\ifx\csname urlprefix\endcsname\relax\def\urlprefix{URL }\fi

\end{thebibliography}


\begin{thebibliography}{40}
\expandafter\ifx\csname natexlab\endcsname\relax\def\natexlab#1{#1}\fi
\expandafter\ifx\csname url\endcsname\relax
  \def\url#1{\texttt{#1}}\fi
\expandafter\ifx\csname urlprefix\endcsname\relax\def\urlprefix{URL }\fi

\bibitem[{Apkarian and Adams(1998)}]{Apkarian:98a}
Apkarian, P., Adams, R., 1998. Advanced gain-scheduling techniques for
  uncertain systems. IEEE Transactions on Automatic Control 6, 21--32.

\bibitem[{Apkarian and Gahinet(1995)}]{Apkarian:95a}
Apkarian, P., Gahinet, P., 1995. A convex characterization of gain-scheduled
  $\mathcal{H}_\infty$ controllers. IEEE Transactions on Automatic Control 5,
  853--864.

\bibitem[{Ariba and Gouaisbaut(2009)}]{Ariba:09}
Ariba, Y., Gouaisbaut, F., 2009. Input-output framework for robust stability of
  time-varying delay systems. In: 48th Conference on Decision and Control.
  Shanghai, China, pp. 274--279.

\bibitem[{Belkoura et~al.(2008)Belkoura, Richard, and Fliess}]{Belkoura:08}
Belkoura, L., Richard, J., Fliess, M., 2008. A convolution approach for delay
  systems identification. In: {IFAC} World Congress. Seoul, Korea.

\bibitem[{Bensoussan et~al.(2006)Bensoussan, Prato, Delfour, and
  Mitter}]{Bensoussan:06}
Bensoussan, A., Prato, G.~D., Delfour, M., Mitter, S., 2006. Representation and
  Control of Infinite Dimensional Systems - $2^{nd}$ Edition. Springer.

\bibitem[{Boyd et~al.(1994)Boyd, Ghaoui, Feron, and Balakrishnan}]{Boyd:94a}
Boyd, S., Ghaoui, L.~E., Feron, E., Balakrishnan, V., 1994. Linear Matrix
  Inequalities in Systems and Control Theory. PA, SIAM, Philadelphia.

\bibitem[{Briat(2008)}]{Briat:08phd}
Briat, C., 2008. Control and observation of {LPV} time-delay systems. Ph.D.
  thesis, Grenoble-INP. Available at \url{{http://www.briat.info/thesis/PhDThesis.pdf}}

\bibitem[{Briat et~al.(2008{\natexlab{a}})Briat, Sename, and Lafay}]{Briat:08a}
Briat, C., Sename, O., Lafay, J., 2008{\natexlab{a}}. Delay-scheduled
  state-feedback design for time-delay systems with time-varying delays. In:
  {IFAC} World Congress, Korea, Seoul.

\bibitem[{Briat et~al.(2008{\natexlab{b}})Briat, Sename, and Lafay}]{Briat:08c}
Briat, C., Sename, O., Lafay, J., 2008{\natexlab{b}}. Parameter dependent
  state-feedback control of {LPV} time delay systems with time varying delays
  using a projection approach. In: {IFAC} World Congress, Korea, Seoul.

\bibitem[{Briat et~al.(2009{\natexlab{a}})Briat, Sename, and Lafay}]{Briat:09e}
Briat, C., Sename, O., Lafay, J., 2009{\natexlab{a}}. Delay-scheduled
  state-feedback design for time-delay systems with time-varying delays - a
  {LPV} approach. Systems and Control Letters 58(9), 664--671.

\bibitem[{Briat et~al.(2009{\natexlab{b}})Briat, Sename, and Lafay}]{Briat:09f}
Briat, C., Sename, O., Lafay, J., 2009{\natexlab{b}}. $\mathcal{H}_\infty$
  delay-scheduled control of of linear systems with time-varying delays. IEEE
  Transactions on Automatic Control 42(8), 2255--2260.

\bibitem[{Briat et~al.(2007)Briat, Sename, and Lafay}]{Briat:07}
Briat, C., Sename, O., Lafay, J.-F., 2007. A {LFT}/$\mathcal{H}_\infty$
  state-feedback design for linear parameter varying time delay systems. In:
  European Control Conference 2007, Kos, Greece.

\bibitem[{Fridman(2006{\natexlab{a}})}]{Fridman:06a}
Fridman, E., 2006{\natexlab{a}}. Descriptor discretized {L}yapunov functional
  method: Analysis and design. IEEE Transactions on Automatic Control 51,
  890--897.

\bibitem[{Fridman(2006{\natexlab{b}})}]{Fridman:06}
Fridman, E., 2006{\natexlab{b}}. Stability of systems with uncertain delays: a
  new 'complete' {L}yapunov-{K}rasovskii {F}unctional. IEEE Transactions on
  Automatic Control 51, 885--890.

\bibitem[{Gahinet and Apkarian(1994)}]{Gahinet:94a}
Gahinet, P., Apkarian, P., 1994. A linear matrix inequality approach to
  $\mathcal{H}_{\infty}$ control. International Journal of Robust and Nonlinear
  Control 4, 421--448.

\bibitem[{Gouaisbaut and Peaucelle(2006)}]{Gouaisbaut:06}
Gouaisbaut, F., Peaucelle, D., 2006. Delay dependent robust stability of time
  delay-systems. In: $5^{th}$ {IFAC} Symposium on Robust Control Design.
  Toulouse, France.

\bibitem[{Gu et~al.(2003)Gu, Kharitonov, and Chen}]{GuKC:03}
Gu, K., Kharitonov, V., Chen, J., 2003. Stability of Time-Delay Systems.
  Birkh{\"a}user.

\bibitem[{Han(2005)}]{Han:05}
Han, Q., 2005. Absolute stability of time-delay systems with sector-bounded
  nonlinearity. Automatica 41, 2171--2176.

\bibitem[{Iv\u{a}nescu et~al.(2000)Iv\u{a}nescu, Dion, Dugard, and
  Niculescu}]{Ivanescu:00}
Iv\u{a}nescu, D., Dion, J., Dugard, L., Niculescu, S., 2000. Dynamical
  compensation for time-delay systems: an {LMI} approach. International Journal
  of Robust and Nonlinear Control 10, 611--628.

\bibitem[{Iwasaki and Shibata(2001)}]{Iwasaki:01a}
Iwasaki, T., Shibata, G., 2001. {LPV} system analysis via quadratic separator
  for uncertain implicit systems. IEEE Transactions on Automatic Control 46,
  1195--1208.

\bibitem[{Kao and Rantzer(2007)}]{Kao:07}
Kao, C., Rantzer, A., 2007. Stability analysis of systems with uncertain
  time-varying delays. Automatica 43, 959--970.

\bibitem[{Kolmanovskii and Myshkis(1999)}]{KolmaMy:99}
Kolmanovskii, V., Myshkis, A., 1999. Introduction to the Theory and
  Applications of Functional Differential Equations. Kluwer Academic
  Publishers, Dordrecht, The Netherlands.

\bibitem[{Louisell(1999)}]{Louisell:99}
Louisell, J., 1999. New examples of quenching in delay differential equations
  having time-varying delay. In: 4th European Control Conference.

\bibitem[{Moon et~al.(2001)Moon, Park, Kwon, and Lee}]{Moon:01}
Moon, Y., Park, P., Kwon, W., Lee, Y., 2001. Delay-dependent robust
  stabilization of uncertain state-delayed systems. International Journal of
  Control 74, 1447--1455.

\bibitem[{Niculescu(2001)}]{Niculescu:01}
Niculescu, S.-I., 2001. Delay effects on stability. A robust control approach.
  Vol. 269. Springer-Verlag: Heidelbeg.

\bibitem[{Papachristodoulou et~al.(2007)Papachristodoulou, Peet, and
  S.-I.Niculescu}]{Papa:07}
Papachristodoulou, A., Peet, M.~M., S.-I.Niculescu, 2007. Stability analysis of
  linear systems with time-varying delays: Delay uncertainty and quenching. In:
  46th Conference on Decision and Control. New Orleans, LA, USA, 2007.

\bibitem[{Scherer(2001)}]{Scherer:01}
Scherer, C.~W., 2001. {LPV} control and full-block multipliers. Automatica 37,
  361--375.

\bibitem[{Sename and Briat(2006)}]{Sename:06}
Sename, O., Briat, C., 2006. Observer-based $\mathcal{H}_\infty$ control for
  time-delay systems: a new {LMI} solution. In: IFAC TDS Conference, L'Aquila,
  Italy.

\bibitem[{Shustin and Fridman(2007)}]{Shustin:07}
Shustin, E., Fridman, E., 2007. On delay-derivative-dependent stability of
  systems with fast-varying delays. Automatica 43, 1649--1655.

\bibitem[{Suplin et~al.(2006)Suplin, Fridman, and Shaked}]{Suplin:06}
Suplin, V., Fridman, E., Shaked, U., 2006. $\mathcal{H}_\infty$ control of
  linear uncertain time-delay systems - a projection approach. IEEE
  Transactions on Automatic Control 51, 680--685.

\bibitem[{Tuan et~al.(2001)Tuan, Apkarian, and Nguyen}]{Tuan:01}
Tuan, H., Apkarian, P., Nguyen, T., 2001. Robust and reduced order filtering:
  new {LMI}-based characterizations and methods. {IEEE} Transactions on Signal
  Processing 49, 2875--2984.

\bibitem[{Tuan et~al.(2003)Tuan, Apkarian, and Nguyen}]{Tuan:03}
Tuan, H., Apkarian, P., Nguyen, T., 2003. Robust filtering for uncertain
  nonlinearly parametrized plants. {IEEE} Transactions on Signal Processing 51,
  1806--1815.

\bibitem[{Verriest et~al.(2002)Verriest, Sename, and Pepe}]{VerriestSP:02}
Verriest, E.~I., Sename, O., Pepe, P., 2002. Robust observer-controller for
  delay-differential systems. In: IEEE Conference on Decision and Control. Las
  Vegas, USA.

\bibitem[{Wu(2001)}]{Wu:01}
Wu, F., 2001. A generalized {LPV} system analysis and control synthesis
  framework. International Journal of Control 74, 745--759.

\bibitem[{Wu and Grigoriadis(2001)}]{Wu:01a}
Wu, F., Grigoriadis, K., 2001. {LPV} systems with parameter-varying time
  delays: analysis and control. Automatica 37, 221--229.

\bibitem[{Xu et~al.(2006)Xu, Lam, and Zhou}]{XuLam:06}
Xu, S., Lam, J., Zhou, Y., 2006. New results on delay-dependent robust
  $\mathcal{H}_\infty$ control for systems with time-varying delays. Automatica
  42(2), 343--348.

\bibitem[{Zhang and Grigoriadis(2005)}]{Zhang:05}
Zhang, F., Grigoriadis, K., 2005. Delay-dependent stability analysis and
  $\mathcal{H}_\infty$ control for state-delayed {LPV} system. In:
  Mediterranean Conference on Control and Automation. pp. 1532--1537.

\bibitem[{Zhang et~al.(2001)Zhang, Knospe, and Tsiotras}]{Zhang:01a}
Zhang, J., Knospe, C., Tsiotras, P., 2001. Stability of time-delay systems:
  Equivalence between {L}yapunov and scaled small-gain conditions. IEEE
  Transactions on Automatic Control 46, 482--486.

\bibitem[{Zhang et~al.(2002)Zhang, Tsiotras, and Knospe}]{Zhang:02a}
Zhang, X., Tsiotras, P., Knospe, C., 2002. Stability analysis of {LPV}
  time-delayed systems. Int. Journal of Control 75, 538--558.

\bibitem[{Zhong(2006)}]{Zhong:06}
Zhong, Q., 2006. Robust Control of Time-Delay Systems. Springer-Verlag, London,
  UK.

\end{thebibliography}

\end{document}